\documentclass[11pt]{article}
\usepackage[utf8]{inputenc}
\usepackage[T1]{fontenc}
\usepackage{lmodern}
\usepackage{geometry}
\usepackage[dvipsnames]{xcolor}
\usepackage[colorlinks=true]{hyperref}
\hypersetup{
    linkcolor=Blue,
    citecolor=Blue,
    urlcolor=Blue
}
\usepackage{amsmath,amssymb,amsthm,amsfonts,mathtools}
\usepackage{graphicx,tikz}
\usetikzlibrary{arrows.meta}
\usepackage{boxedminipage}
\theoremstyle{plain}
\usepackage{complexity}
\usepackage{braket}
\usepackage{physics}
\usepackage{appendix}
\usepackage{authblk}
\usepackage[nameinlink,capitalize]{cleveref}  
\usepackage{aliascnt}
\newcommand{\eps}{\varepsilon}
\newcommand{\dtv}{{\rm d_{TV}}}
\renewcommand{\abs}[1]{\lvert #1 \rvert}

\theoremstyle{plain}
\newtheorem{theorem}{Theorem}[section]
\newaliascnt{lemma}{theorem}
\newtheorem{lemma}[lemma]{Lemma}
\aliascntresetthe{lemma}
\newaliascnt{corollary}{theorem}

\aliascntresetthe{corollary}
\newaliascnt{proposition}{theorem}

\aliascntresetthe{proposition}
\newaliascnt{definition}{theorem}
\newtheorem{definition}[definition]{Definition}
\aliascntresetthe{definition}

\newcommand\emails[1]{\begingroup
\renewcommand\thefootnote{}\footnote{#1}
\addtocounter{footnote}{-1}\endgroup}

\def\ep{\varepsilon}

\def\tr{{\rm tr}}

\def\vev#1{\langle{#1}\rangle}

\def\and{\quad {\rm and} \quad}

\def\CC{{\cal C}}

\def\nn{\nonumber\\}

\def\ketbra#1{ |{#1}\rangle\!\langle{#1}| }

\DeclareMathOperator*{\Ex}{\mathbb{E}}

\def\SU{\mathrm{SU}}

\allowdisplaybreaks

\title{Local random quantum circuits converge to the Porter-Thomas distribution in polynomial depth}

\author[1]{Aniruddha Sen}
\author[2,1]{Nicholas Hunter-Jones}

\affil[1]{Department of Computer Science, University of Texas at Austin, Austin, TX}
\affil[2]{Department of Physics, University of Texas at Austin, Austin, TX}

\date{} 

\begin{document}

\maketitle

\thispagestyle{empty}

\begin{abstract}
Porter-Thomas statistics are a characteristic feature of the output distribution of random quantum states and, more broadly, chaotic quantum many-body systems. Convergence to Porter-Thomas plays a central role in random circuit sampling and experimental demonstrations of quantum advantage, where the output statistics of low-depth random quantum circuits are expected to be approximately Porter-Thomas, despite the absence of a rigorous proof of convergence. We show that the output distribution of polynomial-depth brickwork random circuits converges inverse-polynomially in total variation distance to the Porter-Thomas distribution. Specifically, consider the output probability distribution over a fixed bitstring of a local random quantum circuit, constructed from nearest-neighbor Haar random gates. Then, for any $m \geq 0$, the distribution corresponding to circuits of depth $O(n^{2m+1}\log(n))$ is at most $O(1/n^m)$ far in total variation distance from the Porter-Thomas distribution. Our proof uses moment bounds from approximate designs, analytic estimates for characteristic functions, and a local anticoncentration property for inverse moments. 
\end{abstract}

\emails{\hspace*{-2mm} Email: \href{mailto:aniruddhasen@utexas.edu}{\tt aniruddhasen@utexas.edu}, \href{mailto:nickrhj@utexas.edu}{\tt nickrhj@utexas.edu}}

\section{Introduction}
The goal of quantum advantage is to demonstrate that there are computational tasks that quantum processors can perform more efficiently than classical computers. Random circuit sampling has been studied extensively as an important task to establish a near-term quantum advantage, with both complexity-theoretic and empirical evidence of its hardness for classical computers \cite{boixo2018characterizing,aaronson2016complexity,bouland2019complexity,arute2019quantum,hangleiter2023computational,gao2025establishing}. The task is defined as follows. Given as input a classical description of a random quantum circuit $U$ on $n$ qubits, return samples from its \emph{output} distribution, which is the probability distribution over bitstrings corresponding to the state $U\ket{0^n}$. This is straightforward when given access to a quantum computer, but is widely believed to be a hard problem for classical computers. Quantum supremacy experiments have been based on this task, including experimental demonstrations over the past few years on 50 to 100 noisy qubits \cite{arute2019quantum,wu2021strong,zhu2021quantum,morvan2023phase,decross2025comp,gao2025establishing,abanin2025edge}. In order to \emph{classically verify} the output of the quantum computation, the Linear Cross Entropy Benchmark (LXEB) test has been the standard technique for certifying random circuit sampling \cite{boixo2018characterizing}. We know that the output distribution of Haar random unitaries on $n$ qubits passes this test \cite{kliesch2021theory,hangleiter2023computational}. In particular, we know such output distributions are approximately \emph{Porter-Thomas} distributed. But in practice, we only expect to be able to implement linear or polynomial depth random circuits. This motivates the following question, which had remained open in any rigorous sense despite strong numerical evidence \cite{boixo2018characterizing}.
\begin{center}
    \emph{Is the output distribution of polynomial depth random circuits \\ close to the Porter-Thomas distribution?}
\end{center}
In this work, we answer this question  in the affirmative in a natural setting, by showing that the output distribution of \emph{brickwork} random circuits indeed converges to the Porter-Thomas distribution in polynomial depth, in total variation (or TV) distance.

This result strengthens the LXEB test as a benchmark for random circuit sampling, while also helping us better understand the properties of polynomial depth random quantum circuits. Formally, we prove the following theorem.
\begin{theorem}\label{thm:combined-tv-intro}
Let $b \in \{0,1\}^n$ be a fixed bitstring of length $n$, and let 
\begin{equation}
P\coloneq p_U(b)=\left|\vev{b|U|0}\right|^2 \,,\quad U \sim \mathcal{C}(l) \,,
\end{equation}
be the probability distribution of measuring $b$, where $U$ is sampled from a random brickwork circuit ensemble $\mathcal{C}(l)$ of depth $l$ such that it forms a $d^{-k}$-approximate $k$-design, with $k \ge 2$ and $d=2^n$. Let $Q$ be the corresponding probability distribution where $U$ is sampled from the Haar measure over $\SU(d)$. Then,
\begin{equation}
\label{eq:tv-Haar}
\dtv(P,Q) \le O\left(\sqrt{\frac{\log(k)}{k}}\right) \,.
\end{equation}
In particular, if $k=\Theta(n^{2m}\log(n))$, which implies $l = \Theta(n^{2m+1}\log(n))$ by \cref{thm:brickwork-design}, then
\begin{equation}
    \dtv(P,Q) \le O(1/n^m).
\end{equation}
\end{theorem}

We build upon prior work showing local random quantum circuits form approximate unitary $k$-designs in polynomial depth \cite{brandao2016local,haferkamp2021improved,haferkamp2022random,harrow2018approx,chen2406incompressibility,baer2026random}. Indeed, this design property of random circuits was recently used to show that linear depth random circuits pass the LXEB test \cite{hunter2026ideal}. However, their result did not imply closeness to the Porter-Thomas distribution. This can be explained by the fact that their proof only used limited information about the output distribution of random circuits, which is that a finite number of their moments approximately match the moments of the Porter-Thomas distribution. We can relate this to a more general problem, known as the Hausdorff moment problem \cite{hausdorff1923momentprobleme,akhiezer2020classical} in mathematics, which motivates some of our techniques. The problem asks if a distribution compactly supported in $[0,1]$, corresponding to a given sequence of moments, exists under some conditions. It is known that matching any finite number of moments is insufficient to establish  closeness of distributions, even if they are required to be \emph{well-behaved} in some sense \cite{tardella2001note,tagliani2003note}. So, even if information about moments is sufficient to establish properties like anticoncentration \cite{dalzell2022random} and bounded support of the output distribution \cite{hunter2026ideal}, this presents a barrier to such techniques for showing two distributions are close in a distance measure such as TV distance. This motivates us to specifically analyze brickwork random circuits, using knowledge of the per-gate distribution and circuit architecture in our proof, along with moments. However, we expect our techniques to generalize to other common circuit ensembles that form designs as well.

We briefly discuss some relevant work that has studied properties of the output distribution of random circuits. If we consider the distribution over \emph{bitstrings} of a fixed random circuit (instead of a fixed bitstring over random circuits), then we know that \emph{learning} this distribution in the statistical query model is hard, even at super logarithmic depth. Along with this result, Ref.~\cite{nietner2025average} also shows that the output distribution of a brickwork random quantum circuit is constantly far from any fixed distribution in TV distance with probability exponentially close to one. 

In another line of work, the LXEB test from random circuit sampling has been used in certified randomness protocols \cite{aaronson2023certified, bassirian2026certified} as well, by showing that the output distribution passing the test contains at least $\Omega(n)$ bits of entropy. A version of the protocol was recently demonstrated experimentally \cite{liu2025certified}.
Porter-Thomas statistics also play a key role in recent approaches to benchmarking and characterizing errors from random circuit sampling data \cite{manole2025learn}.

The emergence of Porter–Thomas statistics has also been studied analytically in dual-unitary models \cite{claeys2024fock}, Brownian circuits \cite{bentsen2024comp}, and as a signature for chaotic quantum many-body dynamics \cite{mark2022bench,mark2024max}.

\subsection{Technical Overview}
Our techniques build on a simple observation, which is that the TV distance between two compactly supported probability distributions can be roughly bounded by the $L^2$ norm of the distance between their \emph{characteristic functions}, or Fourier transforms. This is formally stated in \cref{thm:brickwork-bounded}. Let $f$ be the distribution corresponding to polynomial depth random circuits that forms a $k$-design, and let $g$ be the Porter-Thomas distribution. Notice that $f$ and $g$ are supported in the domain $[0,1]$, but their characteristic functions $\hat{f}$ and $\hat{g}$ are supported over all of the complex space $\mathbb{C}$. For intuition, we can just consider them to be supported over the real line $\mathbb{R}$ for now. It suffices to bound the squared absolute difference of their Fourier masses, $\abs{\hat{f}-\hat{g}}^2$, over $\mathbb{R}$. We divide the real line into three parts and bound each part using separate techniques, as explained below. For a more detailed proof, we refer the reader to \cref{section:combined-bound}.

For the first bound, consider expanding the characteristic functions as a Taylor series. This allows us to express the TV distance between the distributions in terms of the differences between their moments. Since $f$ corresponds to a $k$-design, the first $k+1$ moments of $f$ and $g$ are close. This allows us to control the TV distance between them in a \emph{bounded} domain in Fourier space, roughly up to $O(d)$ on the real line, but not entirely, since we only control a finite number of moments. In particular, we prove the following lemma. 
\begin{lemma}[Informal version of \cref{lem:moment-matching-low-frequency}]
\label{lemma:inf-low-moments}
Given the probability density functions $f,g$ as described above, define $z \coloneq f-g$, and let its characteristic function be $\hat{z}$. Then,
\begin{equation}
    \int_{\abs{t}\le d/2}\abs{\hat{z}(t)}^2 \, dt \le  O\left(\frac{d}{k2^{2k}}\right)\,.
\end{equation}
\end{lemma}
What additional information would higher moments have given us? In some sense, higher moments just control the \emph{smoothness} or regularity of a distribution. To this effect, we use the technique of \emph{mollification} to show that there exist distributions $f_{\ast}$ and $g_{\ast}$, close in TV distance to $f$ and $g$ respectively, that are sufficiently smooth. On these mollified distributions, we show that the Fourier mass goes to zero above a certain threshold as proven in the following theorem.
\begin{theorem}[Informal version of \cref{thm:mollification}]
\label{thm:inf-mollification}
Given the probability distribution functions $f,g$ as described above, assume $\|f'\|_{L^1} = ad$, where $d=2^n$ and $a \in \mathbb R$ is a constant. Then, for every integer $k\ge 1$ and every $0<h<1$, there exists a probability density function $f_\ast \in C^\infty$ such that
\begin{equation}
d_{\mathrm{TV}}(f,f_\ast)\le h \,,    
\end{equation}
and the corresponding characteristic function $\widehat{f_\ast}(s)=0$ for all $\abs{s} > a'd/h$ for some constant $a' > a > 0$.
\end{theorem}
Thus, it suffices to bound the TV distance between $f_\ast$ and $g_\ast$. Notice that this theorem assumes $\|f'\|_{L^1} = O(d)$, which is a strict quantitative bound on smoothness we require in order to apply mollification. This is mainly proven in \cref{thm:inverse-moment}, following a reduction from \cref{lem:reduction-TV-to-R}, which reduces bounding $\|f'\|_{L^1}$ to a small ball anticoncentration bound over the probability mass going into a fixed single gate in the last layer of a brickwork random circuit. 

Finally, there is one more technical detail we need to account for. Notice that \cref{thm:inf-mollification} bounds the Fourier mass $\hat{f_\ast}$ only beyond $d/h$ for some $0 < h < 1$. Compare this with the bound up to $d/2$ from \cref{lemma:inf-low-moments}. These domains do not overlap, so \textit{a priori} the Fourier mass may become unbounded in the gap. We use techniques from complex analysis, in particular the Hadamard three-circle theorem (\cref{thm:three-circle}), in order to show that the Fourier mass remains controlled if it is sufficiently well bounded from both sides.

\begin{lemma}[Informal version of \cref{lem:intermediate-fourier-band}]
Given the probability distribution functions $f,g$ as described above, define $z \coloneq f-g$, and let its characteristic function be $\hat{z}$. Then, for any $0 < h \le 1$,
\begin{equation}
    \int_{d/2\le\abs{t}\le d/h} \abs{\hat{z}(t)}^2 \, dt \le \frac{de^{\Omega(-kh^2)}}{h} \,.
\end{equation}
\end{lemma}

These three bounds over different regions of the Fourier space are combined with a certain setting of parameters in \cref{thm:combined-tv} in order to get our final bound on the TV distance between $f$ and $g$.

We note that our techniques appear to not extend to obtaining \emph{inverse-exponential} bounds on TV distance in polynomial depth, even under stronger smoothness assumptions about the distribution of random circuits. However, we conjecture that our bounds are not tight due to numerical evidence suggesting extremely fast convergence of the empirical distribution to the Porter-Thomas distribution \cite{boixo2018characterizing}. It is known that certain random circuit quantities become close to their Haar value at log depth \cite{BF13,dalzell2022random,heinrich2025anti,laracuente2025quantum} and that there are circuit ensembles that form relative-error designs in log depth \cite{schuster2024random,laracuente2024designs}. Thus we leave open the question of whether closeness to Porter-Thomas occurs at some sublinear depth.

Another key open question is whether a similar result can be established in the presence of (sufficiently weak) noise. For random circuits with local depolarizing noise (or any local unital noise), for what depth and noise rate $\ep$ is the output distribution still close to Porter-Thomas? There is a recently understood phase transition in the noise rate \cite{dalzell2024noise,heinrich2022randomized,ware2023sharp,morvan2023phase}. As local noise rapidly drives the output distribution to uniform, a sufficient condition for the noise rate should be $\ep\ll 1/n\ell $ \cite{dalzell2024noise}. If $\ell_{\rm PT}$ is the depth at which the noiseless circuit first becomes approximately Porter-Thomas, is there some depth $\ell_{\rm PT}\leq \ell \ll 1/\ep n$ at which the noisy circuit output is close in TV distance to Porter-Thomas?

\section{Preliminaries}\label{section:prelim}
We consider systems of $n$ qubits and, throughout the text, we denote the total Hilbert space dimension as $d=2^n$.
Consider a unitary $U \in \SU(2^n)$.  Denote the Haar measure over $\SU(d)$ as $\mu_H$, and $b \in \{0,1\}^{n}$ as an arbitrary but fixed bitstring.

\begin{definition}[Output probability distribution]
    We define the output probability distribution $p_U(b)$ for a fixed bitstring $b$ and a unitary $U$ sampled from an ensemble $\nu$ as
    \begin{equation}
        p_U(b)=\left|\vev{b|U|0}\right|^2 \,,\quad U \sim \nu \,.
    \end{equation}
\end{definition}

\begin{definition}[Porter-Thomas distribution] The Porter-Thomas distribution describes the distribution of output probabilities $p$ of Haar random unitaries on $n$ qubits, on the input $0^n$ and fixed output bitstring $b$. It is given by the following expression,
\begin{equation}
    f_{\rm PT}(p) \coloneq d\,e^{-d\, p} \,.
\end{equation}
\end{definition}
The exact expression for the distribution of probabilities $p$ is given by $f_{\rm Haar}(p)=(d-1)(1-p)^{d-2}$ \cite{Zyczkowski_2000,Kus_1988}, i.e.\ a $\mathrm{Beta}(1,d-1)$ random variable. One can check that this converges to the above expression for large $d$. More precisely, the exponential Porter-Thomas and Haar probability densities agree up to $O(1/d)$ corrections and have total variation distance ${\rm d}_{\rm TV}(f_{\rm PT},f_{\rm Haar})\leq O(1/d)$. Thus for convenience, in the rest of this paper we refer to the distribution of output probabilities over Haar random unitaries on $n$ qubits as being the Porter-Thomas distribution.

\begin{definition}[$\eps$-approximate $k$-designs]
A probability distribution $\nu$ on $\SU(d)$ is an $\eps$-approximate $k$-design if and only if the $k$-fold twirl of $\nu$ satisfies
\begin{equation}
    \big\|\mathcal{M}_{\nu}^{(k)} - \mathcal{M}_{\mu_H}^{(k)}\big\|_{\diamond} \le \eps \,,
\end{equation}
where the $k$-fold twirl $\mathcal{M}_{\nu}^{(k)}(\rho) = \int d\nu\, U^{\otimes k}\rho U^{\dagger \otimes k}$, and $\diamond$ indicates the diamond norm distance.
\end{definition}

For convenience, we focus on `brickwork random quantum circuits' \cite{NRVH16,NVH17,vonKey17,RQCstatmech}, which simplifies some aspects of our analysis. Nevertheless, we expect our results to hold for random quantum circuits on any connected architecture with 2-local gates. 

\begin{definition}[Brickwork random quantum circuits]
A brickwork random quantum circuit on $n$ qubits of depth $l$ is constructed as follows. Assume $n$ is even and take periodic boundary conditions. On each odd numbered layer, apply a gate $U_{i,i+1}$ drawn from the Haar measure over $\SU(4)$ on qubits $(i,i+1)$ for odd $i$, $1 \le i \le n$; on even layers apply such gates on qubits $(i,i+1)$ for even $i$, identifying qubits $n+1$ and $1$. Apply these gates for $l$ layers.
\end{definition}

\noindent We will use the following fact about the moments of the output probabilities $p_{U}(b) = |\vev{b|U|0}|^2$ for Haar random unitaries $U\sim \mu_H$ \cite{nietner2025average}
\begin{equation}
\underset{U \sim \mu_H}{\mathbb{E}}\!\left[\abs{p_U(b)}^k \right]
= \frac{k!}{\prod_{i=0}^{k-1} (d+i)}\,.
\end{equation}
For depth $\ell$ random quantum circuits, the moments of the output probabilities match those of Haar random unitaries once the distribution of circuits forms approximate $k$-designs. 

\begin{theorem}[Brickwork circuits are approximate designs \cite{baer2026random}]
\label{thm:brickwork-design}  Let $\mathcal{C}(l)$ be the ensemble of depth $l$ brickwork Haar random circuits on $n$ qubits.
For every $k\ge1$ and $0<\eps<1$, if
\begin{equation}
    l\ge C \bigl(nk+\log(1/\eps)\bigr),
\end{equation}
then $\mathcal{C}(l)$ is an additive $\eps$-approximate unitary $k$-design and $C$ is a constant.
\end{theorem}
This follows from a $k$ independent lower bound on the spectral gap of the moment operator for random circuits.
In particular, we will be using the fact that brickwork Haar random circuits of depth $\ell \ge \Omega(nk)$ form $d^{-k}$-approximate $k$-designs.
We will also use the fact that the distribution is bounded in a finite interval with high probability.

\begin{theorem}[Brickwork circuit distributions are bounded \cite{hunter2026ideal}]
\label{thm:brickwork-bounded}  
Fix a bitstring $b\in\{0,1\}^n$. Let $\mathcal{C}(l)$ be an ensemble of depth $l \ge \Omega(nk)$ brickwork Haar random circuits on $n$ qubits, $k \ge 2$, such that $U \sim \mathcal{C}(l)$. Then,
\begin{equation}
    p_U(b) \leq \frac{2k}{d}
\end{equation}
with probability over the choice of circuit at least $1-2^{-k}$.
\end{theorem}
When $k=\Omega(n)$, one can also union bound over all output probabilities and show that the largest output probability is $\max_b p_U(b)\leq O(n/2^n)$ with probability $>1-2^{-n}$ over circuits.

\begin{proof}
Consider the probability of sampling a fixed bitstring $b\in\{0,1\}^n$. Then, for any integer $k\ge 1$, Markov's inequality gives
\begin{equation}
\Pr_{U\sim C(l)}\!\left[p_U(b)\ge t\right]
\le
\frac{\mathbb{E}_{U\sim C(l)}[p_U(b)^k]}{t^k}.
\end{equation}
Since $C(l)$ is an additive $d^{-k}$-approximate unitary $k$-design from \cref{thm:brickwork-design}, we get the bound
\begin{equation}
\mathbb{E}_{U\sim C(l)}[p_U(b)^k]
\le
\frac{k!}{\prod_{i=0}^{k-1}(d+i)} + d^{-k}
\le
\frac{k!}{d^k} + d^{-k}.
\end{equation}
Therefore,
\begin{equation}
\Pr_{U\sim C(l)}\!\left[p_U(b)\ge t\right]
\le
\frac{k!}{(dt)^k}+\frac{d^{-k}}{t^k} = \frac{k!+1}{(dt)^k}.
\label{eq:output-probability-tail}
\end{equation}
In particular, taking $t=ck/d$ for any constant $c>1$, and using
$k!+1\le k^k$ when $k \ge 2$, gives
\begin{equation}
    \Pr_{U\sim C(l)}\!\left[p_U(b)\ge \frac{ck}{d}\right] \le \left(\frac{kd}{dck}\right)^k = c^{-k}.
\end{equation}
Thus, setting $c=2$ yields
\begin{equation}
    p_U(b) \le \frac{2k}{d}
\end{equation}
with probability $\ge 1-2^{-k}$.
\end{proof}

\subsection{Total variation bound in terms of characteristic functions}
We present the following result concerning the total variation distance between two arbitrary probability distributions $f$ and $g$. We will later set $f$ to be the output probability distribution of random circuits which form $\eps$-approximate $k$-designs and $g$ to be the Porter-Thomas distribution.

\begin{theorem}[\cite{bobkov2016proximity}]
\label{thm:bobkov}
Let $f$ and $g$ be probability density functions with characteristic functions $\hat{f}$ and $\hat{g}$, respectively, concentrated on the interval $[-L, L]$. Then, the following inequality holds.
\begin{equation}
\dtv^2(f,g) \leq \frac{L}{4\pi} \int_{-\infty}^{\infty} |\hat{f}(t) - \hat{g}(t)|^2 \, dt.   
\end{equation}

\end{theorem}

\medskip
\noindent Now, consider the characteristic function $\hat{f}$ of $f$.

\begin{equation}
\hat{f}(t) = \int_{-\infty}^{\infty} e^{itx} \, dF(x) = \sum_{r=0}^{\infty} \frac{(it)^r}{r!} m_r,    
\label{eq:char-taylor}
\end{equation}
where $m_r$ are the moments of the distribution $f$. If we assume the first $k+1$ moments of $f$ are finite, then we also get the following useful expression \cite{knapp2005basic} by the Taylor remainder theorem,
\begin{equation}
\hat{f}(t) = \sum_{r=0}^{k-1} \frac{(it)^r}{r!} m_r + R_k(it) \,,
\end{equation}
where the remainder term $\abs{R_k(it)} \le \frac{\abs{t}^k m_k}{k!}$. Define the difference of the characteristic functions $\hat{z}(t)$ as follows.
\begin{equation}\label{eq:moment_expansion}
\hat{z}(t) = \hat{f}(t) - \hat{g}(t) = \sum_{r=0}^{\infty} \frac{(it)^r}{r!} (m_r - m'_r),    
\end{equation}
where $m'_r$ are the moments of $g$. So, the problem of bounding the total variation distance between two distributions can be expressed as bounding 
\begin{equation}
     \int_{-\infty}^{\infty}\abs{\hat{z}(t)}^2 \, dt \,.
\end{equation}
But what if our distributions are only \emph{mostly} bounded in an interval $[-L,L]$ instead of being fully supported on $[-L,L]$? This turns out to be the case we will care about, so we prove a simple variation of the previous theorem.

\begin{theorem}
\label{cor:bobkov-supported}
Let $f$ and $g$ be probability density functions with characteristic functions $\hat{f}$ and $\hat{g}$, respectively, concentrated on the interval $[0, L]$ with probability mass at least $1-\delta$. Then, the following inequality holds.
\begin{equation}
    \dtv(f,g) \leq \left({\frac{L}{8\pi} \int_{-\infty}^{\infty} |\hat{f}(t) - \hat{g}(t)|^2 \, dt}\right)^{1/2} + \delta \,.
\end{equation}
\end{theorem}
\begin{proof}
    Let $z \coloneq f-g$ and $\hat{z} = \hat{f} - \hat{g}$. By the triangle inequality, Cauchy--Schwarz on $L^2(\mathbb{R})$, and the Plancherel theorem,
\begin{align}
    \dtv(f,g) &=\frac12\int_{0}^L |z(x)|\,dx +\frac12 \int_{x \notin [0,L]}|z(x)|\,dx \\
    &\le \frac{\sqrt{L}}2\|z\|_{L^2(\mathbb R)} + \delta \\
    &=\left(\frac{L}{8\pi}\int_{\mathbb R}|\widehat{z}(t)|^2\,dt
    \right)^{1/2} + \delta \,.
    \label{eq:smoothed-fourier-tv}
\end{align}
\end{proof}

\section{Combined bound}
\label{section:combined-bound}
The proof of our main theorem (\cref{thm:combined-tv-intro}) is based on bounding the quantity $\abs{z(t)}^2$ as suggested by \cref{thm:brickwork-bounded}. In \cref{section:moment-bounds}, we will bound certain regimes of the mass under $\abs{z(t)}^2$ in terms of some parameters. In \cref{section:smoothness}, we obtain a bound on the distance between distributions via smoothing. In this section, we motivate the details of our technical analysis by showing how to combine some key bounds and assumptions, which we prove in later sections, in order to get a global bound for the TV distance between $f$ and $g$.

\begin{theorem}[Restatement of \cref{thm:combined-tv-intro}]\label{thm:combined-tv}
Let $b \in \{0,1\}^n$ be a fixed bitstring of length $n$, and let 
\begin{equation}
P\coloneq p_U(b)=\left|\vev{b|U|0}\right|^2 \,,\quad U \sim \mathcal{C}(l) \,,
\end{equation}
be the probability distribution over measuring $b$, where $U$ is sampled from a random brickwork circuit ensemble $\mathcal{C}(l)$ of depth $l$ such that it forms a $d^{-k}$-approximate $k$-design, with $k \ge 2$ and $d=2^n$. Let $Q$ be the corresponding probability distribution where $U$ is sampled from the Haar measure over $\SU(d)$. Then,
\begin{equation}
\dtv(P,Q) \le O\left(\sqrt{\frac{\log(k)}{k}}\right) \,.
\end{equation}
In particular, if $k=\Theta(n^{2m}\log(n))$, which implies $l = \Theta(n^{2m+1}\log(n))$ by \cref{thm:brickwork-design}, then
\begin{equation}
   \dtv(P,Q) \le O(1/n^m)\,.
\end{equation}
\end{theorem}

\begin{proof}
Here, let $f$ denote the PDF of $P$ and $g$ be the PDF of $Q$. From \cref{cor:bobkov-supported}, we know that if two probability density functions (PDFs) $f,g$ are concentrated in some interval $[0,L]$ with probability at least $1-\delta$, their TV distance is upper bounded by
\begin{equation}
\sqrt{\frac{L}{8\pi}\int_{-\infty}^{\infty}\abs{\hat{z}(t)}^2 \, dt} + \delta\,,
\end{equation}
where $\hat{z}(t)=\int_{-\infty}^{\infty} z(x) e^{itx} dx$ is the characteristic function of $z(x) \coloneq (f-g)(x)$.

Our broad strategy will be to divide the integral into 3 parts and bound each part separately using different techniques as follows, for some $0 < h < 1$.
\begin{itemize}
    \item $0 \le \abs{t} \le d/2$: This region of the integral is bounded by using the fact that the first $k$ moments of $f$ and $g$ are $\eps$-close, as shown in \cref{lem:moment-matching-low-frequency}.
    \item $\abs{t} \ge d/h$: This region of the integral is bounded by showing there exist well-bounded \textit{mollified} distributions that are obtained through smoothing the tails of $\hat{f}$ and $\hat{g}$, as shown in \cref{thm:mollification}
    \item $d/2 \le \abs{t} \le d/h$: This region of the integral is bounded by an application of the Hadamard three-circle theorem, as defined in \cref{thm:three-circle}, which intuitively says that for analytic functions satisfying the above two properties, the intermediate region must also be bounded, as shown in \cref{lem:intermediate-fourier-band}.
\end{itemize}

Now, we know that the first $k$ moments of $f$ and $g$ are $\eps$-close, since $f$ comes from an $\eps$-approximate $k$-design. Thus, we directly apply \cref{lem:moment-matching-low-frequency} for some $k \ge 2$ and $\eps \le d^{-k}$, and set $L=2k/d$ as the circuit has depth $\Omega(nk)$. So,

\begin{equation}
    \label{eq:low-frequency-control}
    \frac{L}{8\pi}\int_{\abs{t} \le d/2}\abs{\hat{z}(t)}^2 \, dt \le O(1/2^{2k}) \,.
\end{equation}

Since moments do not control the shape or regularity of the distributions themselves, we need additional tools to bound the distance between them. So, we smoothen both distributions using mollification as described in \cref{thm:mollification}, with a common mollifier, obtaining the mollified functions ${f_{\ast}}$ and ${g_\ast}$ which are $\eps'$ apart in TV distance from $f$ and $g$ respectively. We also obtain the additional property that $\widehat{f_\ast}(t)=\widehat{g_\ast}(t)=0$ for some $\abs{t} \ge a'd/\eps'$, for some constant $a'$, and for any choice of $\eps'>0$. Now, by the triangle inequality over TV distance,
 \begin{equation}\label{eq:tri-mol}
     \dtv(f,g) \le \dtv(f_\ast,g_\ast) + 2\eps'
 \end{equation}
So, we will now just bound the distance between the mollified functions. Let $z_\ast \coloneq f_\ast - g_\ast$, so we get $\abs{\widehat{z_\ast}} \le \abs{\widehat{z}}$ by property of convolution, and since the mollifier is non-negative and integrates to 1. Both the mollified functions are now concentrated in $[0,L]$ where $L=2k/d$ from \cref{thm:brickwork-bounded}, with at most $\eps' + \delta$ probability mass outside, since they are $\eps'$-close in TV distance to the non-mollified functions, which already had a $\delta$ tail mass. Thus, the bound from Eq.~\eqref{eq:low-frequency-control} becomes

\begin{equation}
    \frac{L}{8\pi}\int_{\abs{t} \le d/2}\abs{\widehat{z_\ast}(t)}^2 \, dt \le O(1/2^{2k}) \,.
\end{equation}
Next, by the additional property of mollification which cuts off the tails of $\widehat{f_\ast}$ and $\widehat{g_\ast}$,

\begin{equation}
    \frac{L}{8\pi}\int_{\abs{t} \ge d/h}\abs{\widehat{z_\ast}(t)}^2 \, dt = 0 \,
\end{equation}
where $h \coloneq \eps'/a'$.

\noindent Finally, to bound the intermediate region of the integral, we apply \cref{lem:intermediate-fourier-band} to $f$ and $g$ and get the bound
\begin{equation}
    \frac{L}{8\pi}\int_{d/2\le\abs{t}\le d/h} \abs{\hat{z_\ast}(t)}^2 \, dt \le \frac{L}{8\pi}\int_{d/2\le\abs{t}\le d/h} \abs{\hat{z}(t)}^2 \, dt \le \frac{ke^{\Omega(-kh^2)}}{h} \,.
\end{equation}

since $\abs{\widehat{z_\ast}} \le \abs{\widehat{z}}$ by property of convolution. Adding these bounds together along with the triangle inequality, substituting the value of $h$, and setting $L=2k/d$, so $\delta=2^{-k}$, by \cref{thm:brickwork-bounded}, we obtain a bound in TV distance in terms of $\eps'$, applying Eq.~\eqref{eq:tri-mol} and \cref{cor:bobkov-supported},

\begin{equation}
    \dtv(f,g) \le O\left(\sqrt{\frac{1}{2^{2k}} + \frac{ke^{\Omega(-k\eps'^2)}}{\eps'}} + 2^{-k} + 3\eps'\right)
\end{equation}

Setting $\eps'=O(\sqrt{\log(k)/k})$ to minimize this expression with a sufficiently large constant factor, we obtain the overall bound

\begin{equation}
    \dtv(f,g) \le O\left(\frac{\sqrt{k}}{2^k} + 2^{-k} + 3\sqrt{\frac{\log(k)}{k}}\right) \le  O\left(\sqrt{\frac{\log(k)}{k}}\right) \,.
\end{equation}
for $k \ge 2$.
\end{proof}

This tells us that brickwork random quantum circuits of depth $O(n^{2m+1}\log(n))$ are at most $O(1/n^m)$ far in TV distance from the Porter-Thomas distribution, under some assumptions about the smoothness and bounds of the former distribution. In order to complete this proof, we now need to prove the lemmas referred to in this section.

\section{Bounds from moments} \label{section:moment-bounds}
\begin{lemma}
\label{lem:moment-matching-low-frequency}
Given two probability density functions $f,g$ supported in $[0,1]$, let the function $z(x) \coloneq f(x)-g(x)$, and its characteristic function
\begin{equation}
    \hat{z}(t)=\int_{-\infty}^{\infty} (f(x)-g(x))e^{itx} \, dx \,,
\end{equation}
Assume that the first $k+1$ moments of $f,g$ are $m_j,m_j'$ such that $\abs{m_j-m_j'} \le \eps$, and $m_j \le j!/d^j$, for $0 \le j \le k$, where $0< \eps \le d^{-k}$, and making the mild assumption that $d \ge 4k$. Then,
\begin{equation}
    \int_{\abs{t}\le d/2}\abs{\hat{z}(t)}^2 \, dt \le  O\left(\frac{d}{k2^{2k}}\right)\,.
\end{equation}
\end{lemma}
\begin{proof}
Since we assumed the first $k+1$ moments of $f$ and $g$ are $\eps$-close, we define
\begin{equation}
\Delta m_r = (m_r - m'_r) \,, \quad \abs{\Delta m_r} \le \eps    
\end{equation}
for $r = 0$ to $k$. 
Substituting the above bounds into Eq.~\eqref{eq:moment_expansion} from Theorem \ref{thm:bobkov}, and then applying Taylor's remainder theorem with the remainder terms $R_k(it)$, $R_k'(it)$ we get,
\begin{align}
\abs{\hat{z}(t)} &\le \left|\sum_{r=0}^{k-1} \frac{(it)^r}{r!} \Delta              m_r + \sum_{r=k}^{\infty} \frac{(it)^r}{r!} m_r -                     \sum_{r=k}^{\infty} \frac{(it)^r}{r!} m_r'\right|\\
           &\le \sum_{r=0}^{k-1} \frac{\abs{t}^r}{r!} \eps + \abs{R_k(it)} + \abs{R_k'(it)} \, \\
           &\le \sum_{r=0}^{k-1} \frac{\abs{t}^r}{r!} \eps + \frac{\abs{t}^k(m_k + m_k')}{k!} \label{eq:zt_exp} \,,
\end{align}
applying the triangle inequality. Now, for the first summation, we use the bound
\begin{equation}
S_1(t) \coloneqq \sum_{r=0}^{k-1} \frac{\abs{t}^r}{r!}
\le
\begin{cases}
e^{\abs{t}}, & \abs{t} \le 2k, \\[4pt]
2\frac{\abs{t}^{k-1}}{(k-1)!}, & \abs{t} > 2k \,.
\end{cases}
\label{eq:s1t}
\end{equation}
Substituting back into Eq.~\eqref{eq:zt_exp} and using the given upper bounds on $m_k, m_k'$, we obtain 
\begin{align}
|\hat{z}(t)| &\leq S_1(t)\cdot\eps + \frac{|t|^{k}}{k!} \cdot                  3\frac{k!}{d^k} \\
       &\leq S_1(t)\cdot\eps + \frac{3|t|^{k}}{d^k}
\end{align}
Integrating over a symmetric interval $[-T, T]$, assuming $T \ge 2k$ and substituting $S_1(t)$ from Eq.~\eqref{eq:s1t} accordingly,
\begin{align}
\int_{-T}^{T} |\hat{z}(t)|^2 \, dt
&\le \int_{-T}^{T} \left(S_1(t)\cdot\eps+\frac{3|t|^{k}}{d^{k}}\right)^2 \, dt \\
&\le 4\int_{0}^{T} \left(S_1(t)^2\cdot\eps^2+ \frac{9t^{2k}}{d^{2k}}\right) \, dt \\
&\le 4\ep^2 \int_{0}^{2k} e^{2t} \, dt
   + \frac{16\eps^2}{(k-1)!^2}\int_{2k}^{T} |t|^{2k-2} \, dt
   + \frac{36}{d^{2k}} \int_{0}^{T} |t|^{2k} \, dt \\
&\le 2\eps^2(e^{4k}-1)
   + \frac{16\eps^2 T^{2k-1}}{(2k-1)(k-1)!^2}
   + \frac{36T^{2k+1}}{d^{2k}(2k+1)} .
\end{align}

\noindent Now, if we let $T=d/2$ and $\eps \le d^{-k}$, then
\begin{equation}
\int_{-d/2}^{d/2} |\hat{z}(t)|^2 dt \le \frac{2e^{4k}}{d^{2k}} + \frac{16}{(2k-1)(k-1)!^2d^{2k}}\left(\frac{d}{2}\right)^{2k-1} + \frac{18d}{(2k+1) \, 2^{2k}} \le O\left(\frac{d}{k2^{2k}}\right) .
\end{equation}
Note that this assumes $d/2 \ge 2k$. However, for the case $k \ge d/4$, we could set $T$ to be larger, and still obtain the fact that the resulting integral is well-bounded.
\end{proof}

\noindent Next, we will prove that $\int_{d/2}^{d/h} |z(t)|^2\,dt$ and $\int_{-d/h}^{-d/2} |z(t)|^2\,dt$ are negligible quantities, for some $0<h\leq 1$. This again uses information about the moments of the distributions. The main technical ingredient of our proof, however, will be from complex analysis, in particular it involves an application of the Hadamard three-circle theorem.

\begin{theorem}[Hadamard three-circle theorem]
\label{thm:three-circle}
Let the function $H$ be analytic on a neighborhood of the annulus $r_1\le |z|\le r_2$, where $0<r_1<r_2$, and define
\begin{equation}
    M(r):=\max_{|z|=r}|H(z)|.
\end{equation}
Then $\log M(r)$ is convex as a function of $\log r$. In particular, if
$r_1\le r\le r_2$, then
\begin{equation}
    M(r)\le M(r_1)^\theta M(r_2)^{1-\theta},
\qquad
\theta=
\frac{\log(r_2/r)}{\log(r_2/r_1)}.
\end{equation}
\end{theorem}

\begin{lemma}
\label{lem:intermediate-fourier-band}
Given two probability density functions $f,g$ supported in $[0,1]$, let the function $z(x) \coloneq f(x)-g(x)$, and its characteristic function
\begin{equation}
    \hat{z}(t)=\int_{-\infty}^{\infty} (f(x)-g(x))e^{itx} \, dx ,
\end{equation}
Assume that $f,g \in C^1((0,1)) \cap C([0,1])$, and that the first $k+1$ moments of $f,g$ are $m_j,m_j'$ such that $\abs{m_j-m_j'} \le \eps$, and $m_j \le j!/d^j$, for $1 \le j \le k$, where $0< \eps \le d^{-k}$, $d \ge 2$. Then, for any $\sqrt{\gamma/k} \le h \le 1$, for some some constant $\gamma>0$,
\begin{equation}
    \int_{d/2\le\abs{t}\le d/h} \abs{\hat{z}(t)}^2 \, dt \le \frac{de^{\Omega(-kh^2)}}{h} \,.
\end{equation}
\end{lemma}

\begin{proof}
Our proof mainly involves an application of the Hadamard three-circle theorem (HTCT). We first define the function $H(s)$ on which we want to apply HTCT. Let $L > 0$, then
\begin{equation}
    H(s) \coloneq \int_{-L}^{L} (f(x)-g(x))e^{-sx} \, dx ,
\end{equation}
Note that $H$ is analytic, and that $\abs{H(-it)} \rightarrow \abs{\hat{z}(t)}$ as $L \rightarrow \infty$. This is going to be the intermediate `circle'. Now, we first consider the inner circle, which is given by $\abs{s} \le s_0d$ for some constant $s_0 > 0$. We can bound $\abs{H(s)}$ in this regime by using \cref{lem:fourier-distance-from-moments} with $s_0$ sufficiently small such that $\abs{H(s)} \le 2^{-c_0'k}$, setting $\eps = d^{-k}$. 

For the outer circle, we first need to prove a tail bound for $f,g$. Applying Markov's inequality on the random variable $X$ corresponding to $f(x)$,
\begin{align}
    \Pr[X \ge x] &= \Pr[X^j \ge x^j] \le \frac{\mathbb{E}[X^j]}{x^j} \\
    &= \frac{m_j}{x^j} \le \frac{j!}{(xd)^j} \le \left(\frac{j}{xd}\right)^j \\
    &\le e^{-xd\ln(2)/2} \,,
\end{align}
if we choose $j=\lfloor xd/2\rfloor $, assuming $x \le 2k/d$. Since $x \le 2k/d$ then $k \ge xd/2$, and thus $j = \lfloor xd/2\rfloor \le k$ is indeed a valid choice of $j$. We will later set $L=2k/d$, so we do not require the bound for $x > 2k/d$ at the moment.

We apply this result to the expansion of $\mathbb{E}[e^{aX}]$ (and $\mathbb{E}[e^{aY}]$), following an identity which holds for any non-negative random variable $X$ (restricted to $[0,L]$),
\begin{align}
\mathbb{E}[e^{aX}] &= 1 + a\int_{0}^{L} e^{ax} \Pr[X \ge x] \,dx \\
   &\le 1 + a\int_{0}^{L} e^{ax} e^{-xd\ln(2)/2} \,dx \\
   &\le 1 + \frac{a}{d\ln(2)/2-a}
\end{align}
assuming $a \le d\ln(2)/2$.
Using this result, we can show $H(s)$ is bounded if we assume $-a=\Re{s} \ge -s_1d$, where $0 < s_1 < \ln(2)/2$, since
\begin{align}
 \left|H(s)\right| &= \left|\mathbb{E}\left[e^{-sX}\right] - \mathbb{E}                            \left[e^{-sY}\right]\right| \\
                   &\le \left|\mathbb{E}\left[e^{-sX}\right]\right| + \left|\mathbb{E}\left[e^{-sY}\right]\right| \\
                   &\le \mathbb{E}\left[e^{-\Re{s}X}\right] + \mathbb{E}\left[e^{-\Re{s}Y}\right] \\
                   &\le 2 + \frac{2s_1}{\ln(2)/2-s_1} = O(1)
\end{align}
which is a constant, so let this quantity be $c_1$.
Now, we would like to apply HTCT with this condition on $s$, but $\Re{s} \ge -s_1d$ is a half-plane rather than a disk. So, we define the function \begin{equation}
w=\phi(s)\coloneq \frac{s}{s+2s_1d}, \qquad s=\phi^{-1}(w)=\frac{2s_1dw}{1-w},
\end{equation}
which maps the half-plane $\Re{s} \ge -s_1d$ to the unit disk $|w| \le 1$. Define
\begin{equation}
\overline{H}(w)\coloneq H\left(\frac{2s_1dw}{1-w}\right).
\end{equation}
Now, from the half-plane bound on $H(s)$, we get a bound on the outer circle over $w$ as follows.
\begin{equation}
\abs{\overline{H}(w)} \le c_1, \qquad |w|<1.
\end{equation}
Now, we transfer the inner circle bound from $H$ to $\overline{H}$. Recall that
\begin{equation}
    |H(s)|\le 2^{-c_0'k}, \qquad |s|\le s_0d.
\end{equation}
Set
\begin{equation}
r_0\coloneq \frac{s_0}{2s_1+s_0}.
\end{equation}
Then, if $|w|\le r_0$,
\begin{align}
\left|\phi^{-1}(w)\right| &= \left|\frac{2s_1dw}{1-w}\right| 
\le \frac{2s_1d|w|}{1-|w|} \le \frac{2s_1dr_0}{1-r_0} \\
&= s_0d.
\end{align}
Hence, we get the inner circle bound over $w$ bounded by some constant $r_0$ as
\begin{equation}
\abs{\overline{H}(w)}\le 2^{-c_0'k}, \qquad |w|\le r_0.
\end{equation}

Now, we are ready to apply HTCT to the function $\overline{H}$. Note that since the function is conformal, the analyticity of $H$ carries over to $\overline{H}$. So, for $r_0\le |w|<1$, we obtain the following upper bound from \cref{thm:three-circle} for any intermediate circle.
\begin{equation}
\abs{\overline{H}(w)} \le \left(2^{-c_0'k}\right)^{\theta(w)} \cdot c_1^{1-\theta(w)},
\end{equation}
where $\theta(w) = \log(1/|w|)(\log(1/r_0))^{-1}$. Thus, for some constants $\gamma,c_1' > 0$,
\begin{equation}
\abs{\overline{H}(w)} \le \gamma2^{-c_1'k\theta(w)}.
\end{equation}

Applying the inverse map $\phi^{-1}$ to obtain an upper bound corresponding to $H$, we get $\abs{H(s)} \le \gamma2^{-c_1'k\theta(\phi(s))}$. Now, we can directly bound $\abs{H(-it)}$. We first calculate $\abs{\phi(-it)}$ and lower bound $\theta(-it)$ to simplify the bound.
\begin{equation}
\abs{\phi(-it)} = \left|\frac{-it}{-it+2s_1d}\right| = \frac{|t|}{\sqrt{t^2+4s_1^2d^2}}.
\end{equation}
Then,
\begin{align}
\log\frac{1}{\abs{\phi(-it)}} &= \frac{1}{2}\log\left(1+\frac{4s_1^2d^2}{t^2}                                     \right) \\
                              &\ge \frac{c_2}{1+(t/d)^2}.
\end{align}
where $c_2 = 2s_1^2$. Since $r_0$ is also a constant, this then implies
\begin{equation}
\theta(-it) \ge \frac{c}{1+(t/d)^2} \,
\end{equation}
for some fixed constant $c$.
Thus, 
\begin{equation}
|H(-it)| \le \gamma2^{-ck/(1+(t/d)^2)} \,. \label{eq:65}
\end{equation}
By using Markov's inequality, we bound the approximation error between $\hat{z}(t)$ and $H(-it)$ as
\begin{align}
    \abs{\hat{z}(t)-H(-it)} &= \left|\int_L^\infty z(x)e^{itx} \,dx\right| \le \int_L^\infty \abs{z(x)} \,dx \le \Pr[X > L] + \Pr[Y > L] \\
    &\le \frac{m_k+m_k'}{L^k} \le \frac{3k!}{d^k} \cdot \frac{d^k}{(2k)^k} \le 3\cdot 2^{-k} \le 3e^{-k\ln(2)} 
\end{align}
where $X \sim f(x), Y \sim g(x)$, using $k! \le k^k$, and setting $L=2k/d$. Finally, integrating Eq.~\eqref{eq:65} squared and applying triangle inequality with $\abs{\hat{z}(t)}^2$ and the above expression proves our lemma.
\end{proof}
 
\section{Smoothing distributions}\label{section:smoothness}
Here, we show that there exists a \emph{smoothed} distribution which is close in TV distance to the original distribution, which quantitatively only depends on the $L_1$ norm of its first derivative.

\begin{theorem}[Smoothing the output distribution]
\label{thm:mollification}
Let $C$ be a brickwork Haar random circuit on $n$ qubits corresponding to the unitary $U$.  Let $b\in\{0,1\}^n$ be an arbitrary output string, and define $F \coloneq \left|\langle b|U|0^n\rangle\right|^2$. Let $f$ denote the probability density function of $F$, such that it is monotone decreasing and absolutely continuous, and assume $\|f'\|_{L^1} = a\, d$, where $d=2^n$ and $a \in \mathbb R$ is a constant. Then, for every integer $k\ge 1$ and every $0<\varepsilon<1$, there exists a probability density function $f_\ast \in C^\infty$ such that
\begin{equation}
d_{\mathrm{TV}}(f,f_\ast)\le \varepsilon \,,    
\end{equation}
and the corresponding characteristic function 
\begin{equation}
\widehat{f_\ast}(s)=0    
\end{equation}
for all $\abs{s} > a'd/\eps$ for some constant $a' > a > 0$.
\end{theorem}

\begin{proof}
Let $\rho\in C^\infty$ be a fixed nonnegative function known as a mollifier such that $\int_{\mathbb R}\rho(t)\,dt=1$, $\int_{\mathbb R}\abs{t}\rho(t)\,dt = b \le 1$, and its Fourier transform $\widehat{\rho}(s) = 0$ for $\abs{s} \ge c$ for some constant $c > 0$.  Such functions indeed exist, for instance we can check that the function $J(x)=\frac{3}{2\pi}\left(\frac{\sin(x)}{x}\right)^4$ satisfies these properties, having $b=3\ln(2)/\pi$ and $c=4$. Now, for some $h \in \mathbb{R}^+$, let
\begin{equation}
\rho_h(t)=h^{-1}\rho(t/h) \,,    
\end{equation}
and consider the convolution of $f$ with this mollifier which we denote as
\begin{equation}
f_h=f*\rho_h \,.    
\end{equation}
We can view the function $f_h$ as a `smoothed' form of $f$. Since $\rho_h\ge0$ and $\int_{\mathbb R}\rho_h=1$, the function $f_h$ is nonnegative and satisfies
\begin{equation}
\int_{\mathbb{R}} f_h(p)\,dp=\int_{\mathbb{R}} \int_{\mathbb{R}}\rho_h(t)f(p-t) \,dt \,dp= \left(\int_{\mathbb{R}}\rho_h(t)\,dt\right)\left(\int_{\mathbb{R}}f(p-t) \,dp\right)=1.
\end{equation}
where the second equality follows from Fubini's theorem, the change of variables $u=p-t$, and integrating over $du$. We also know that convolution with a mollifier $\rho_h \in C^{\infty}$ results in a $C^{\infty}$ function. Thus, $f_h\in C^\infty(\mathbb{R})$ is a valid probability density function.

Next, we bound the TV distance between $f$ and $f_h$. Since $f$ is absolutely continuous and $f'\in L^1(0,1)$, we have the bound
\begin{equation}
\int_0^1 |f(p-t)-f(p)|\,dp
\le
|t|\,\|f'\|_{L^1}    
\end{equation}
for all $t\in\mathbb R$. Therefore
\begin{align}
\|f_h-f\|_{L^1(0,1)} &= \int_{0}^1|f_h(p)-f(p)| \,dp \\
    &= \int_0^1\left|\int_{\mathbb R}\rho_h(t)f(p-t)\,dt-f(p)\right| \,dp \\
    &= \int_0^1\left|\int_{\mathbb R}\rho_h(t)\bigl(f(p-t)-f(p)\bigr)\,dt\right| \,dp \\
    &\le \int_{\mathbb R}\rho_h(t)\left(\int_0^1 |f(p-t)-f(p)|\,dp\right) \,dt \label{eq:77} \\
    &\le \|f'\|_{L^1}\int_{\mathbb R}|t|\rho_h(t) \,dt. \\
    &= \|f'\|_{L^1}\int_{\mathbb R}|t|\,h^{-1}\rho(t/h) \,dt \label{eq:79} \\
    &= \|f'\|_{L^1}h\int_{\mathbb R}|u|\rho(u) \,du \\
    &= hb\|f'\|_{L^1} \,,
\end{align}
where we substitute $u=t/h$ in Eq.~\eqref{eq:79}, and the last line follows from the definition of $\rho$ and $b$.  Now, since $d_{\mathrm{TV}}(f,f_h)=\frac{1}{2}\|f-f_h\|_{L^1}$, it suffices to choose
\begin{equation}
h = \frac{2\eps} {b\|f'\|_{L^1}} \,,   
\end{equation}
which implies
\begin{equation}
d_{\mathrm{TV}}(f,f_h)\le \eps.  
\end{equation}

Now, we show that the characteristic function $\widehat{f_h}$ is bounded according to this value of $h$. By the scaling and convolution identities for the Fourier transform and the definition of $\rho_h(t)$, we know $\widehat{\rho_h}(s)=\widehat{\rho}(hs)$ and $\widehat{f_h}(s)=\widehat{f}(s)\widehat{\rho}(hs)$. Thus, $\widehat{\rho}(s)=0 \implies \widehat{\rho_h}(s/h)=0 \implies \widehat{f_h}(s/h)=0$. 
Since $\widehat\rho(s)=0$ whenever $|s|>c$, $\widehat{f_h}(s)=0$ for $\abs{s}>c/h$.
With the above choice of $h$,
\begin{equation}
 \widehat{f_h}(s)=0\,
\qquad
\text{for }
\abs{s}>\frac{abcd}{2\varepsilon} \,,
\end{equation}
so we define $a'\coloneq abc/2 > 0$, which is a constant independent of $n$ and $\varepsilon$. In particular, for $J(x)$, $a'=6a\ln(2)/\pi \approx 1.32a$. This completes our proof.
\end{proof}

The remaining assumption we need to prove is that $\|f'\|_{L^1} \le O(d)$ where $f$ is the output probability distribution of a polynomial depth random quantum circuit that forms a $k$-design. We first reduce it to proving an inverse moment bound as follows.
\begin{lemma}[Total variation bound]
\label{lem:reduction-TV-to-R}
Let $U \sim \mathcal{C}(l)$ be a brickwork random circuit on $n$ qubits. Let $d=2^n$, and consider a fixed output string $b\in\{0,1\}^n$ and its corresponding probability $F=\left|\langle b|U|0^n\rangle\right|^2$. Let $f$ denote the density of the random variable $F$ over a randomly sampled circuit. Then
\begin{equation}
\|f'\|_{L^1} \le O(d) \,,
\end{equation}
assuming an inverse moment bound i.e. $\mathbb{E}[R^{-1}] \le O(d)$, as proved in \cref{thm:inverse-moment}, where $R$ is the random variable defined in Eq.~\eqref{eq:r-def}.
\end{lemma}
\begin{proof}
Let
\begin{equation}
U = G_*C \,   
\end{equation}
where $G_*\in \SU(4)$ is the final two-qubit Haar random gate on qubits $(0,1)$, and $C$ corresponds the remaining circuit. Consider the marginal probability over all outputs of $G_*$, i.e.
\begin{equation}\label{eq:r-def}
R=\sum_{z \in\{0,1\}^2}\left|\langle z,\bar{b}|C|0^n\rangle\right|^2 .   
\end{equation}
where $\bar{b}$ denotes the value of $b$ outside of qubits $(0,1)$. Conditioned on $C$, the four amplitudes entering $G_\ast$ form a fixed vector $v\in\mathbb C^4$ with $\|v\|_2^2=R$. We know the output distribution of a Haar random unitary from $\SU(4)$ follows a $\operatorname{Beta}(1,3)$ distribution. Therefore
\begin{equation}
F_{G_\ast}\,|\,R = R B ,
\qquad
B\sim \operatorname{Beta}(1,3) \,,   
\end{equation}
where
\begin{equation}
f_B(b)=3(1-b)^2,\qquad 0<b<1 \,. 
\end{equation}
Conditioning on $R=r$, the random variable $F$ is a rescaled Beta variable.
Hence, the conditional density of $F$ is
\begin{equation}
f_{F\mid R=r}(p) = \frac1r\,f_B\!\left(\frac pr\right) =\frac3r\left(1-\frac pr\right)^2 \,,\qquad (0<p<r).    
\end{equation}
since the cumulative density is $\Pr(P \le p) = \Pr(rb \le p) = \Pr(b \le p/r) = F_B(p/r)$, and differentiating gives us the probability density. Now, by the definition of conditional probability, we can average over all $R$ to get the overall probability density
\begin{equation}
f(p) = \mathbb{E}\!\left[f_{F\mid R}(p)\right] = 3\mathbb{E}\!\left[\frac1R\left(1-\frac pR\right)^2 \right] \,,\qquad (0<p<r).  \label{eq:dist_expR}
\end{equation}
Differentiating under the expectation, we get
\begin{equation}
f'(p) = -6\mathbb{E}\!\left[\frac1{R^2} \left(1-\frac pR\right) \right] \,,\qquad (0<p<r).    
\end{equation}
The differentiation under the expectation is justified because, on the event
$\{R>p\}$, $\frac1{R^2}\left|1-\frac pR\right|\le\frac1{R^2}\le\frac1{p^2},$ which is integrable for each fixed $p>0$.
Now, moving the expectation outside the integral
\begin{align}
\|f'\|_{L^1} 
    &= 6\int_0^1\mathbb{E}\left[\frac1{R^2}\left(1-\frac pR\right)\right]\,dp \,,\qquad (0<p<r). \\
    &=6\mathbb{E}\left[\int_0^R\frac1{R^2}\left(1-\frac pR\right)\,dp\right] \\
    &=6\mathbb{E}\left[\frac1{R^2}\cdot \frac R2\right] \\
    &=3\mathbb{E}[R^{-1}]. \\
    &= O(d)\,.
\end{align}

\end{proof}
\noindent It only remains to prove that $\mathbb{E}[R^{-1}] \le O(d)$.

\section{Inverse moment bound}

\begin{theorem}[Inverse moment bound]
\label{thm:inverse-moment}
Let $U \sim \mathcal{C}(l)$ be a brickwork random circuit on $n$ qubits such that it is drawn from a $d^{-k}$-approximate $k$-design. Factorize $U=G_*C$ such that $G_*$ is a fixed gate in the final layer of the circuit, and let $S_*$ be the pair of output qubits of $G_*$. We define
\begin{equation}
R\coloneq \sum_{z\in\{0,1\}^{S_*}} |\langle z,x|C|0^n\rangle|^2,    
\end{equation}
where $x$ is a fixed string on $\overline{S_*}$. Let $d=2^n$ and $k\ge 8$. Then, there exists a constant $K>0$, such that
\begin{equation}
\mathbb E[R^{-1}] \le Kd \,.
\end{equation}
\end{theorem}

\begin{proof}
We show that the probability mass $R$ is sufficiently anticoncentrated according to the following bound
\begin{equation}
    \Pr\!\left[R\le \frac{t}{d}\right]\le At^{12/11}, \qquad 0<t\le 1 \,,
    \label{eq:inverse-block-small-ball}
\end{equation}
for some constant $A>0$. The desired inverse moment bound on $R$ will then follow by expressing it as an integral which uses this bound.

The proof broadly analyzes how the output probability mass changes, with respect to the size of a contiguous interval of qubits, considering only the \emph{last two layers} of the random circuit, and starting with a fixed gate $G_*$. If we can show it doesn't increase too much as we increase the size of the interval, we can then use the design property of the circuit to lower bound the mass on a larger interval. Combining these two facts will allow us to lower bound the probability mass $R$ going into $G_*$. We first start by understanding how the mass locally increases as we add two consecutive gates $G$ and $H$ in the circuit over the last two layers, and then iterate.

\begin{figure}[h]
  \centering
  \resizebox{0.95\linewidth}{!}{
  \begin{tikzpicture}[
    x=1cm,y=1cm,font=\small,
    qwire/.style={draw=black!45,line width=.6pt},
    qfree/.style={draw=black!45,line width=1.2pt},
    tracedout/.style={densely dashed,draw=black!45,line width=.6pt},
    output/.style={densely dashed,draw=black!50,line width=.6pt},
    box/.style={draw,fill=white,line width=.7pt,
                minimum height=.7cm,inner sep=4pt,align=center,rounded corners},
    gate/.style={box,minimum width=2.05cm,minimum height=.8cm},
    ghostgate/.style={gate,draw=black!35,text=black!45,densely dotted},
    cancel/.style={draw=black!45,line width=.7pt},
    >={Latex[length=2mm]}
  ]

  \draw[->] (-.5,.65) -- (-.5,3.95)
    node[midway,rotate=90,above,font=\footnotesize] {circuit time};
  \draw[output] (0,4.35) -- (14.1,4.35);
  \node[fill=white,inner sep=4pt,align=center,font=\footnotesize]
    at (7.05,4.35) {final outputs\\(fixed level)};

  \node[align=center] at (2.75,5.65)
    {\textbf{Before:} $S_0=S_*=\{q_3,q_4\}$\\[2pt]
     $R_0=R,\qquad r_0=4$};
  \node[font=\footnotesize] at (1.45,4.95) {fixed to $x_1x_2$};
  \node[font=\footnotesize] at (4.05,4.95) {already free};

  \foreach \x in {.8,2.1,3.4}
    \draw[qwire] (\x,.45) -- (\x,4.35);
  \draw[qfree] (3.4,2.15) -- (3.4,4.35);

  \draw[tracedout] (4.7,.45) -- (4.7,4.35);

  \node[box,minimum width=3.7cm] at (2.1,.1)
    {$\rho$};

  \node[gate] (Gbefore) at (2.75,1.75)
    {$G_{23}$\\[-1pt]{\footnotesize preceding layer}};
  \node[gate] (Hbefore) at (1.45,2.95)
    {$H_{12}$\\[-1pt]{\footnotesize final layer}};
  \node[ghostgate] (starbefore) at (4.05,2.95)
    {$G_*$\\[-1pt]{\footnotesize already free}};

  \draw[cancel]
    ([xshift=-2pt,yshift=-2pt]starbefore.south west) --
    ([xshift=2pt,yshift=2pt]starbefore.north east);

  \node[box,minimum width=2.15cm] at (1.45,4.35)
    {$|x_1x_2\rangle\langle x_1x_2|$};
  \node[box,minimum width=2.15cm] at (4.05,4.35)
    {$I_3\otimes I_4$};

  \foreach \x/\lab in {.8/1,2.1/2,3.4/3,4.7/4}
    \node[fill=white,inner sep=2pt,font=\footnotesize]
      at (\x,.85) {$\lab$};

  \draw[->,line width=.9pt] (5.75,2.5) -- (8.35,2.5)
    node[midway,above=6pt,align=center] {free $\{q_1,q_2\}$}
    node[midway,below=6pt,align=center] {enlarge $S_0$ to $S_1$};

  \begin{scope}[xshift=8.6cm]
    \node[align=center] at (2.75,5.65)
      {\textbf{After:} $S_1=\{q_1,q_2,q_3,q_4\}$\\[2pt]
       $R_1=s,\qquad r_1=4r_0=16$};
    \node[font=\footnotesize] at (2.75,4.95)
      {all outputs in $S_1$ free};

    \foreach \x in {.8,2.1,3.4}
      \draw[qfree] (\x,.45) -- (\x,4.35);
    \draw[tracedout] (4.7,.45) -- (4.7,4.35);

    \node[box,minimum width=3.7cm] at (2.1,.1)
      {$\rho$};

    \node[ghostgate] (Gafter) at (2.75,1.75)
      {$G_{23}$\\[-1pt]{\footnotesize preceding layer}};
    \node[ghostgate] (Hafter) at (1.45,2.95)
      {$H_{12}$\\[-1pt]{\footnotesize final layer}};
    \node[ghostgate] (starafter) at (4.05,2.95)
      {$G_*$\\[-1pt]{\footnotesize already free}};

    \foreach \g in {Gafter,Hafter,starafter}
      \draw[cancel]
        ([xshift=-2pt,yshift=-2pt]\g.south west) --
        ([xshift=2pt,yshift=2pt]\g.north east);

    \node[box,minimum width=5.1cm,fill=black!5] at (2.75,4.35)
      {$I_1\otimes I_2\otimes I_3\otimes I_4$};

    \foreach \x/\lab in {.8/1,2.1/2,3.4/3,4.7/4}
      \node[fill=white,inner sep=2pt,font=\footnotesize]
        at (\x,.85) {$\lab$};
  \end{scope}

  \node at (2.75,-.85) {$R_0=sP_\rho$};
  \node at (11.35,-.85) {$R_1=s$};

  \node at (7.05,-1.55) {$
    P_\rho=\operatorname{Tr}\!\left[
      (|x_1x_2\rangle\langle x_1x_2|\otimes I_3)
      H_{12}G_{23}\rho G_{23}^{\dagger}H_{12}^{\dagger}
    \right]
  $};
  \end{tikzpicture}
  }
\caption{\textbf{Bounding the probability mass over larger intervals.}
Initially $S_*=S_0=\{q_3,q_4\}$ is free, so $G_*$ contributes to both masses. The neighboring qubits $q_1,q_2$ are fixed until $S_0$ is enlarged to $S_1=\{q_1,q_2,q_3,q_4\}$. This involves inverting the effect of applying $H_{12}$ and $G_{23}$ on the qubits, or \emph{freeing} the fixed qubits. $R_0$ and $R_1$ are the probability masses on $S_0$ and $S_1$ respectively. We only analyze the effect of the gates in the last two layers of the circuit, freeing the qubits from right to left in this figure.}
\label{fig:boundary-peeling}
\end{figure}
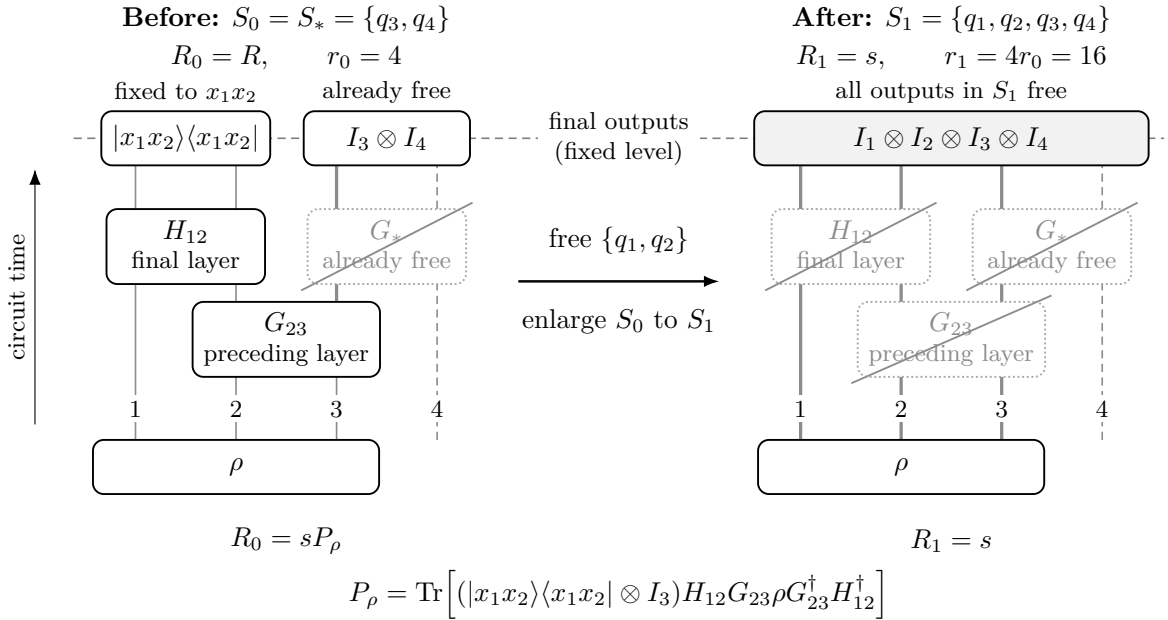

Consider three consecutive qubits $q_1,q_2,q_3$ in the circuit $C$,
and let $G_{23},H_{12}$ be Haar random gates applied in order as
in \cref{fig:boundary-peeling}. Let $x_1,x_2\in\{0,1\}$, and let $\rho$ be the initial density matrix on qubits $q_1,q_2,q_3$. Define
\begin{equation}
P_\rho\coloneq\operatorname{Tr}\!\left[(\ket{x_1x_2}\bra{x_1x_2}\otimes I_3) H_{12}G_{23}\rho G_{23}^{\dagger}H_{12}^{\dagger} \right].
\end{equation}
In other words, $P_\rho\in[0,1]$ is the probability of measuring
the fixed two-bit string $x_1x_2$ on qubits $q_1,q_2$, when $q_3$ (and $q_4$) is left unmeasured.

Now, let $S_i$ be the set of qubits we consider in the $i$-th step. We start by setting $S_0=S_*$, which contains the qubits corresponding to the initial gate $G_*$. So in \cref{fig:boundary-peeling}, $S_0=\{q_3,q_4\}$. We also define the probability mass on the set $S_i$ analogously as
\begin{equation}
R_{i} \coloneq \sum_{z\in\{0,1\}^{S_i}}\bigl|\langle z,x_{\overline S_i}|C|0^n\rangle\bigr|^2.    
\end{equation}
Thus, $R_{0}=R$ and $R_{n}=1$. We also get the useful relation that
\begin{equation}
R_{i+1}=\frac{R_{i}}{P_\rho} \,,  
\end{equation}
where $P_\rho$ is defined by assigning qubit $q_3$ to $S_i$ and qubits $q_1,q_2$ to $S_{i+1}$. In \cref{fig:boundary-peeling}, we will set $S_1 = \{q_1,q_2,q_3,q_4\}$ for instance. So, the probability mass does not decrease as we increase the size of the interval. We increase the number of qubits and gates considered in $S_i$ by two in each step. Note that $R_{i+1}=0 \implies R_{i}=0$.

Instead of bounding the probability mass itself, we bound the inverse expected moment of $R_m$ as we iterate over larger $S_i$ by applying \cref{app:two-gate-beta} as follows. For $1 < q < 2$ and $\varepsilon>0$, we analyze the mass with respect to the bounded convex function
\begin{equation}
     \varphi_{\varepsilon,q}(p) \coloneq (p+\varepsilon)^{-q},\qquad p\in[0,1].   
\end{equation}
Bounding the expectation of this function as we increase the size of the set $S_i$ will allow us to get a good inverse moment bound for some interval $S_i$ which may not necessarily go up to $S_n$ . Letting $\varepsilon\rightarrow 0$ by monotone convergence, $ \varphi_{\varepsilon,q}(p) \rightarrow p^{-q}$, so $\varphi_{\varepsilon,q}(P_{\rho}) = P_\rho^{-q}$. Applying \cref{app:two-gate-beta},
\begin{equation}
    \mathbb E_{G,H}[P_\rho^{-q}]
    \le \bigl(\mathbb E[B^{-q}]\bigr)^2=M_q^2,
    \qquad B\sim\operatorname{Beta}(2,2),
    \label{eq:block-local-negative-moment}
\end{equation}
where
\begin{equation}
    M_q:=\mathbb E[B^{-q}]
    =6\int_0^1 b^{1-q}(1-b)\,db
    =\frac{6}{(2-q)(3-q)}.
    \label{eq:block-beta-negative-moment}
\end{equation}
A similar lemma can be proved when expanding the interval by a
single qubit instead of two, by a similar argument involving
Jensen's inequality to reduce to the pure-state case, followed by
unitary invariance of the Haar measure. We mainly expand by two
qubits in our analysis since expanding by only one output of a single Haar
random gate instead gives probability mass analogous to a
$\operatorname{Beta}(2,2)$ variable, whose inverse expected value
becomes bounded by $M_q$ instead of $M_q^2$. 

These finite negative-moment bounds imply that the probability mass remains positive almost surely, since we can start from $R_n=1$ (which is known by definition) and proceed backwards, which implies all $R_i$'s are positive almost surely. The ratios below are therefore well defined. So for each step, conditioning on the event $\mathcal F_i$ which fixes the randomness over all the gates except $G,H$, by Eq.~\eqref{eq:block-local-negative-moment},
\begin{equation}
    \mathbb{E}_{G,H}[P_\rho^{-q}] = \mathbb E\!\left[\left(\frac{R_{i+1}}{R_i}\right)^q\middle|\mathcal F_i\right] \le M_q^{2}.
    \label{eq:block-conditional-ratio}
\end{equation}

Let us iterate this bound till some index $m$. Consider the ratio of the masses at the step $i$. We get the following bound by expanding $\mathbb{E}_{G,H}[P_\rho^{-q}]$ by the law of total expectation or the tower property.
\begin{align}
    \mathbb E\!\left[\left(\frac{R_m}{R_i}\right)^q\right]&=\mathbb E\!\left[\left(\frac{R_m}{R_{i+1}}\right)^q\mathbb E\!\left[ \left(\frac{R_{i+1}}{R_i}\right)^q\middle|\mathcal F_i\right] \right]\\
    &\le M_q^{2}\mathbb E\!\left[\left(\frac{R_m}{R_{i+1}}\right)^q\right].
\end{align}
Since $\mathcal F_i$ conditions on the gates except $G,H$, the ratio $(R_m/R_{i+1})$ can be taken outside the expectation, and then we apply the bound from Eq.~\eqref{eq:block-conditional-ratio}. Now, iterating from $i=0$ to $i=m-1$, we get
\begin{equation}
    \mathbb E\!\left[\left(\frac{R_m}{R}\right)^q\right]
    \le M_q^{|S_m|-2} \,.
    \label{eq:block-general-ratio-moment}
\end{equation}
Note that we initially started with two free qubits in $S_*$. If the number of qubits $n$ is odd, we may need to apply the iteration at the last step with $M_q$ instead of $M_q^2$; our bounds would remain the same. For the rest of the proof, we set $q=3/2$. By Eq.~\eqref{eq:block-beta-negative-moment}, $M_{3/2}=8$. Let $r_m \coloneq 2^{|S_m|}$, then
\begin{equation}
    \mathbb E\!\left[ \left(\frac{R_m}{R}\right)^{3/2} \right]\le8^{|S_m|-2} =\left(\frac{r_m}{4}\right)^3.
    \label{eq:block-ratio-moment}
\end{equation}

\noindent We now show that the mass on the larger interval $S_m$ is unlikely to be much smaller than the Haar value, by appealing to the design property of $C$. In particular, we claim that for any even $k$
\begin{equation}
    \Pr\!\left[R_m<\frac{r_m}{2d}\right]\le\frac{2^k k!}{r_m^{k/2}}\,.
    \label{eq:block-large-window-tail}
\end{equation}
The proof of the above claim is as follows.

First, note that for fixed $x$, $R_m = \sum_{z\in\{0,1\}^{S_i}} |\vev{z,x_{\overline S_i}|C|0^n}|^2$ can be rewritten as
\begin{equation}
    R_m = \sum_{z\in\{0,1\}^{S_m}} \!\!\tr\Big( C\ketbra{0^n} C^\dagger \ketbra{z,x_{\overline S_m}}\Big) = \tr\Big( C\ketbra{0^n} C^\dagger \Pi_m \Big)\,,
\end{equation}
where $\Pi_m:= \sum_{z} \ketbra{z,x_{\overline S_m}}$ is a projector with rank $r_m = 2^{|S_m|}$. For Haar random $C\sim \mu_H$ the expected value of $R_m$ is $\Ex_{\mu_H} [R_m] = \frac{r_m}{d}$, which follows from a first moment calculation. We now center the random variable as
\begin{equation}
    R_m - \frac{r_m}{d} = \tr\Big( C\ketbra{0^n} C^\dagger \Big(\Pi_m - \frac{r_m}{d} I\Big) \Big)\,,
\end{equation}
such that $\Ex_{\mu_H}[R_m - \frac{r_m}{d}]=0$. For $C\sim \mu_H$, we compute $k$-th moments of the centered quantity as
\begin{equation}
    \Ex_{\mu_H} \bigg[\Big(R_m-\frac{r_m}{d}\Big)^k\bigg] = \tr\bigg(\Ex_{\mu_H}\Big[\big(C\ketbra{0^n}C^\dagger\big)^{\otimes k}\Big] \Big(\Pi_m - \frac{r_m}{d} I\Big)^{\otimes k} \bigg) \,.
    \label{eq:Rmoments}
\end{equation}
For $C\sim \mu_H$, the distribution on the space of states induced by $\ket\psi=C\ket{0^n}$ is the Haar measure on states. The $k$-th moments of Haar random states can be written in terms of permutations on the $k$-fold system \cite{harrow2013church}
\begin{equation}
    \Ex_{\mu_H} \big[ \ketbra{\psi}^{\otimes k}\big] =  \frac{1}{\prod_{i=0}^{k-1}(d+i)} \sum_{\sigma\in S_k} P_\sigma\,,
\end{equation}
where $S_k$ is the symmetric group in $k$ elements and $P_\sigma$ is a permutation operator on the $k$-fold Hilbert space. Proceeding, we have
\begin{align}
    \Ex_{\mu_H} \bigg[\Big(R_m-\frac{r_m}{d}\Big)^k\bigg] &= \frac{1}{\prod_{i=0}^{k-1}(d+i)} \sum_{\sigma\in S_k} \tr\bigg(P_\sigma \Big(\Pi_m - \frac{r_m}{d} I\Big)^{\otimes k} \bigg)\\
    &= \frac{1}{\prod_{i=0}^{k-1}(d+i)} \sum_{\sigma\in S_k} \prod_i^{\ell(\sigma)} \tr\bigg(\Big(\Pi_m - \frac{r_m}{d} I\Big)^{\lambda_i} \bigg)\,.
\end{align}
where for every permutation $\sigma\in S_k$, let $\lambda(\sigma)$ be the cycle type of $\sigma$, an integer partition $\lambda \vdash k $ given as $\lambda = (\lambda_1, \lambda_2, \ldots, \lambda_{\ell(\sigma)})$, and where $\ell(\sigma)$ is the number of cycles in $\sigma$. Note that
\begin{equation}
    \tr\Big(\Pi_m - \frac{r_m}{d}I\Big) = 0\,,
\end{equation}
as $\Pi_m$ has rank $r_m$. Therefore, every permutation $\sigma\in S_k$ with a fixed point, i.e.\ some $\lambda_i=1$, has zero contribution in the sum. It also follows that for any $k\geq 2$
\begin{equation}
    \tr\Big(\big(\Pi_m - \frac{r_m}{d} I\big)^k\Big)\leq \tr\Big(\big(\Pi_m - \frac{r_m}{d} I\big)^2\Big) = r_m - \frac{r_m^2}{d}\leq r_m
\end{equation}
where the first inequality follows from the fact that $\Pi_m - \frac{r_m}{d} I$ is an operator with eigenvalues $\xi$ between $-1$ and $1$, as $\Pi_m$ is a projector and $r_m\leq d$, and thus $\xi^k<\xi^2$ for $|\xi|\leq 1$.
Letting $D_k = \{ \sigma\in S_k ~|~ \sigma(i)\neq i ~\text{for all}~ i\}$ be the set of all permutations with no fixed points, we find
\begin{equation}
    \Ex_{\mu_H} \bigg[\Big(R_m-\frac{r_m}{d}\Big)^k\bigg] \leq \frac{1}{\prod_{i=0}^{k-1}(d+i)} \sum_{\sigma\in D_k} r_m^{\ell(\sigma)} \leq \frac{1}{d^k}\frac{k!}{2} r_m^{k/2}\,,
\end{equation}
as $\ell(\sigma)\leq k/2$ for a permutation with no fixed points. For even $k$, the largest number of cycles in a permutation with no fixed points is the permutation consisting of $k/2$ swaps. The inequality then follows from $\prod_{i=0}^{k-1}(d+i)\geq d^k$ and $|D_k|=\lfloor (k!+1)/e\rfloor \leq k!/2$.

For a $1/d^k$-approximate $k$-design $\CC(\ell)$, using Eq.~\eqref{eq:Rmoments}
\begin{align}
    \bigg|\Ex_{\CC(\ell)} \bigg[\Big(R_m-\frac{r_m}{d}\Big)^k\bigg] - \Ex_{\mu_H} \bigg[\Big(R_m-\frac{r_m}{d}\Big)^k\bigg]\bigg| &= \bigg|\tr\bigg(\Big(\mathcal{M}^{(k)}_{\CC(\ell)}-\mathcal{M}^{(k)}_{\mu_H} \Big) \big(\ketbra{0^n}^{\otimes k}\big) \Big(\Pi_m - \frac{r_m}{d} I\Big)^{\otimes k} \bigg) \bigg|\nn
    &\leq \bigg\|\Big(\mathcal{M}^{(k)}_{\CC(\ell)}-\mathcal{M}^{(k)}_{\mu_H} \Big) \big(\ketbra{0^n}^{\otimes k}\big) \Big(\Pi_m - \frac{r_m}{d} I\Big)^{\otimes k} \bigg\|_1\nn
    &\leq \Big\|\Big(\mathcal{M}^{(k)}_{\CC(\ell)}-\mathcal{M}^{(k)}_{\mu_H} \Big) \big(\ketbra{0^n}^{\otimes k}\big)\Big\|_1 ~\Big\|\Pi_m - \frac{r_m}{d} I \Big\|_\infty^k \nn
    &\leq \frac{1}{d^k}\,,
\end{align}
where in the third line we use H\"older's inequality and in the last line use the definition of the diamond norm and the fact that the eigenvalues $\xi$ of $\Pi_m - \frac{r_m}{d} I$ are $|\xi|\leq 1$. Therefore, for a $1/d^k$-approximate $k$-design $\CC(\ell)$ the centered moments of $R_m$ for even $k$ obey
\begin{equation}
    \Ex_{\CC(\ell)} \bigg[\Big(R_m-\frac{r_m}{d}\Big)^k\bigg]  \leq \frac{1}{d^k}\frac{k!}{2} r_m^{k/2} + \frac{1}{d^k}\leq \frac{k!}{d^k}r_m^{k/2} \,.
    \label{eq:Rmdesign}
\end{equation}
Returning to the probability that $R_m$ deviates away from its expected value $r_m/d$, we bound the probability by centering and using Markov's inequality as
\begin{align}
    \Pr \left[R_m<\frac{r_m}{2d}\right]
    &\le \Pr \left[\Big|R_m-\frac{r_m}{d}\Big|\ge\frac{r_m}{2d}\right] = \Pr \left[\Big|R_m-\frac{r_m}{d}\Big|^k \ge \frac{r_m^k}{(2d)^k}\right]\\
    &\le \frac{(2d)^k}{r_m^k}\Ex_{\CC(\ell)} \bigg[\Big(R_m-\frac{r_m}{d}\Big)^k\bigg]\\
    &\le \frac{2^k k!}{r_m^{k/2}}\,,
\end{align}
where we again take $k$ to be even and use Eq.~\eqref{eq:Rmdesign}.
This proves Eq.~\eqref{eq:block-large-window-tail}.

Now, we just have to combine the above bound with Eq.~\eqref{eq:block-ratio-moment} in order to prove the small ball anticoncentration bound from Eq.~\eqref{eq:inverse-block-small-ball}. 

Let $0<t\le1$, and choose an index $m$ below. If $R \le t/d$, then either $R_m < r_m/2d$ or $R_m/R\ge r_m/2t$. Thus, by applying a union bound, substituting Eq.~\eqref{eq:block-large-window-tail}, and by Markov's inequality on Eq.~\eqref{eq:block-ratio-moment},
\begin{align}
    \Pr[R \le t/d] &\le \Pr\!\left[R_m<\frac{r_m}{2d}\right]
       +\Pr\!\left[\frac{R_m}{R}\ge\frac{r_m}{2t}\right] \\
    &\le \frac{2^k k!}{r_m^{k/2}} + \left(\frac{2t}{r_m}\right)^{3/2}
        \mathbb E\!\left[\left(\frac{R_m}{R}\right)^{3/2}\right] \\
    &\le \frac{2^k k!}{r_m^{k/2}} + \left(\frac{2t}{r_m}\right)^{3/2} \left(\frac{r_m}{4}\right)^3 = \frac{2^k k!}{r_m^{k/2}} + \frac{\sqrt2}{32}\,t^{3/2}r_m^{3/2}.
\label{eq:block-scale-bound}
\end{align}
Here the first term controls a tail bound from the anticoncentration of the $k$-design, while the second term controls the loss in probability mass between $S_m$ and $S_*$.

We choose $r_m=O(t^{-3/11})$ and $k=8$ to balance the two terms. We divide the possible values $t$ can take into two cases, and bound each one separately. If $4t^{-3/11}\le d$, take the first index $m$ for which
$r=2^{|S_m|}\ge4t^{-3/11}$. Since each step adds at most
two qubits,
\begin{equation}
    4t^{-3/11}\le r_m \le16t^{-3/11}.    
\end{equation}
Then, by Eq.~\eqref{eq:block-scale-bound},
\begin{equation}
\Pr[R\le t/d] \le \left(\frac{2^8 8!}{4^4}+2\sqrt2\right)t^{12/11}.  
\end{equation}
For the other case, if $4t^{-3/11}>d$, let $m=L$. Then $R_L=1$, so applying Markov's inequality directly to Eq.~\eqref{eq:block-ratio-moment}
gives
\begin{align}
    \Pr[R\le t/d] &=\Pr\!\left[\frac1R\ge\frac dt\right] \le\left(\frac td\right)^{3/2}
      \mathbb E[R^{-3/2}] \le\left(\frac td\right)^{3/2} \left(\frac d4\right)^3 \\
    &=\frac1{64}t^{3/2}d^{3/2} \le\frac18t^{12/11}.
\end{align}
Consequently, $\Pr[dR\le t]\le At^{12/11}$ for some constant $A>0$ and all $0<t\le1$, proving Eq.~\eqref{eq:inverse-block-small-ball}.

Finally, integrating the lower tail, substituting $t=d/y$, and then applying the bound from Eq.~\eqref{eq:inverse-block-small-ball} gives
\begin{align}
    \mathbb E[R^{-1}]
    &=\int_0^\infty\Pr[R^{-1}>y]\,dy\\
    &\le d+\int_d^\infty\Pr[R<1/y]\,dy\\
    &=d+d\int_0^1\frac{\Pr[dR<t]}{t^2}\,dt\\
    &\le d\left(1+A\int_0^1t^{-10/11}\,dt\right)\\
    &=(1+11A)d = O(d) \,.
\end{align}
This completes the proof.
\end{proof}

\section*{Acknowledgments}
We would like to thank Shivan Mittal for helpful discussions. NHJ and AS acknowledge support in part from DOE grant DE-SC0025615. AS is supported by a TQI Graduate Fellowship.

\section*{Statement of AI Use}
Generative AI provided some of the main ideas (and figure) in the proof of \cref{thm:inverse-moment}, and the initial ideas for proofs in the Appendix, over successive iterations through which we developed and refined the proof techniques. It was also useful in simplifying the proof of \cref{lem:intermediate-fourier-band}, in particular suggesting the application of the Hadamard three-circle theorem. We have verified and edited all proofs for correctness and clarity.

\appendix

\section{Bound from low moments}
\begin{lemma}
\label{lem:fourier-distance-from-moments}
Given two probability density functions $f,g$ supported in $[0,1]$, let
\begin{equation}
H(s)\coloneq \int_{-L}^{L}(f(x)-g(x))e^{-sx}\,dx \,,    
\end{equation}
where $L=2k/d$. Assume that the first $k$ moments of $f,g$ are $m_j,m_j'$ such that $\abs{m_j-m_j'} \le \eps$, and $m_j \le j!/d^j$, for $0 \le j \le k$, where $0< \eps \le d^{-k}$, $d \ge 2$. Then, there exists a constant $c_0>0$ such that, for $k \ge 2$ and any $s\in\mathbb C$ such that $|s|\le d/8$,
\begin{equation}
\abs{H(s)} \le 2^{-c_0k} \,.
\end{equation}
\end{lemma}

\begin{proof}
This proof follows similarly to the proof of \cref{lem:moment-matching-low-frequency}. Since we assumed the first $k+1$ moments of $f$ and $g$ are $\eps$-close, we define $z(x) \coloneq f(x)-g(x)$ and 
\begin{equation}
\Delta m_r \coloneq (m_r - m'_r) \,, \quad \abs{\Delta m_r} \le \eps    
\end{equation}
for $r = 0$ to $k$. So, we can bound $\abs{H(s)}$ through a Taylor series expansion, similar to Eq.~\eqref{eq:char-taylor}. Let us first define the truncated moments as
$\mu_r\coloneq\int_{-L}^L x^r z(x)\,dx$, and the absolute moments of $z(x)$ as $M_k\coloneq\int_0^1 x^k\abs{z(x)}\,dx$. By the Taylor series expansion of $e^{-sx}$, we can expand
\begin{align}
  H(s) &= \int_{-L}^L z(x)e^{-sx}\,dx
        = \sum_{r=0}^{\infty}\frac{(-s)^r}{r!}\,\mu_r \\
       &= \sum_{r=0}^{k-1}\frac{(-s)^r}{r!}\,\Delta m_r + R_k(s),
\end{align}
where we define the remainder term as
\begin{equation}
R_k(s)\coloneq\sum_{r=0}^{k-1}\frac{(-s)^r}{r!}\,(\mu_r-\Delta m_r) + \sum_{r=k}^{\infty}\frac{(-s)^r}{r!}\,\mu_r = T_1(s) + T_2(s) \,.    
\label{eq:rem-Rks}
\end{equation}
We now show a common bound for the absolute values of both $T_1(s)$ and $T_2(s)$. Notice that for $r<k$ we know $\mu_r-\Delta m_r=-\int_L^1 x^r z(x)\,dx$, and for $r\ge k$
we have $\mu_r=\int_0^L x^r z(x)\,dx$. In both cases $x^r\le L^{r-k}x^k$ on
the range of integration (since $x\ge L$ when $r<k$, and $x\le L$ when
$r\ge k$). Hence, using the triangle inequality on Eq.~\eqref{eq:rem-Rks}, and then the bound $\abs{s}L\le k/4$,
\begin{equation}
\abs{R_k(s)} \le \abs{T_1(s)} + \abs{T_2(s)} = \frac{M_k}{L^k}\sum_{r=0}^{\infty}\frac{(\abs{s}L)^r}{r!}
  =e^{\abs{s}L}\,\frac{M_k}{L^k}\le e^{k/4}\,\frac{M_k}{L^k}.    
\end{equation}
Now, notice that we can bound the moments of $z(x)$ by the sum $M_k\le m_k+m_k'\le 2m_k+\eps\le 3k!/d^k$. We also know $L=2k/d$, $\eps \le d^{-k}$, and $\abs{s} \le d/8$, so substituting into Eq.~\eqref{eq:rem-Rks}, we get
\begin{equation}
\abs{H(s)}\le\eps\sum_{r=0}^{k-1}\frac{\abs{s}^r}{r!}+\abs{R_k(s)} \le \sum_{r=0}^{k-1}\frac{d^{r-k}}{8^rr!} + e^{k/4}\cdot\frac{3k!}{d^k}\cdot\frac{d^k}{(2k)^{k}} \le 2^{-c_0k} \,,
\end{equation}
for some constant $0 < c_0 \le 1/4$.
\end{proof}

\begin{lemma}[Two-gate bound]\label{app:two-gate-beta}
Let $q_1,q_2,q_3$ be three consecutive qubits. Let $G,H$ be two-qubit gates drawn independently from the Haar measure on $\mathrm{SU}(4)$, and consider a circuit $C$ such that $G$ first acts on the pair of qubits $(q_2,q_3)$ and then $H$ acts on the pair of qubits $(q_1,q_2)$. Fix $x_1,x_2\in\{0,1\}$, and let $\rho$ be the initial state on qubits $q_1,q_2,q_3$. Define
\begin{equation}
P_\rho \coloneq \Tr\!\left[(\ket{x_1x_2}\!\bra{x_1x_2}\otimes I_3) HG\rho G^{\dagger}H^{\dagger} \right].
\end{equation}
Then for every bounded convex function $\varphi:[0,1]\to\mathbb{R}$
and any $\rho$,
\begin{equation}
\mathbb{E}_{G,H}[\varphi(P_\rho)] \le \mathbb{E}_{B_1,B_2}[\varphi(B_1 \cdot B_2)],
\qquad
B_1,B_2 \overset{\mathrm{i.i.d.}}{\sim} \mathrm{Beta}(2,2), 
\end{equation}
where $\mathrm{Beta}(2,2)$ is a random variable with the density function $6p(1-p)$ on $[0,1]$.
\end{lemma}
\begin{proof}
By Jensen's inequality on the function $\varphi$, we just need to consider the case where $\rho$ is a pure state since by linearity
$\varphi(P_\rho) \le \sum_i p_i\, \varphi(P_{\ketbra{\psi_i}})$ and thus $\mathbb{E}_{G,H}[\varphi(P_\rho)] \le \sum_i p_i\, \mathbb{E}_{G,H}[\varphi(P_{\ketbra{\psi_i}})]$. So we simply need to prove an upper bound for an arbitrary pure state input. In the following, let $\rho = \ketbra{\psi}$ be an arbitrary pure state. The rest of the proof proceeds by using the left/right invariance of the Haar measure acting on the respective mixed state subsystems on qubits $q_1,q_2,q_3$ of $\ket{\psi}$.

We first average over the gate $H$, which is applied second. So we fix the gate $G$ and let $\ket{\Phi}\coloneq G\ket{\psi}$. Since $q_3$ is a single qubit, $\ket{\Phi}$ has a Schmidt decomposition across the partition $[(1,2)\,|\,(3)]$ of the form
\begin{equation}
\ket{\Phi}=\sqrt{\lambda_1}\,\ket{e_1}_{(1,2)}\ket{f_1}_{(3)}
+\sqrt{\lambda_2}\,\ket{e_2}_{(1,2)}\ket{f_2}_{(3)} \,,   
\end{equation}
such that $\lambda_1+\lambda_2=1$, $\ket{e_1},\ket{e_2}\in\mathbb{C}^4$, and $\ket{f_1},\ket{f_2}$ forms an orthonormal basis of $\mathbb{C}^2$.

Now, notice that the state $\ket{w}\coloneq H^\dagger\ket{x_1x_2}$ is a Haar random state in $\mathbb{C}^4$, since $H$ is Haar random. So, we rewrite
\begin{align}
P_{\ket{\psi}\bra{\psi}} &= \Tr\!\left[(\ket{x_1x_2}\bra{x_1x_2}\otimes I_3)          H\ket{\Phi}\bra{\Phi}H^{\dagger} \right] \\
         &= \bra{\Phi}H^\dagger\bigl(\ket{x_1x_2}\bra{x_1x_2}\otimes I_3\bigr)H\ket{\Phi} \\
         &=\bra{\Phi}\bigl(\ket{w}\bra{w}\otimes I_3\bigr)\ket{\Phi} \\
         &=\bigl\|(\bra{w}\otimes I_3)\ket{\Phi}\bigr\|^2 \,, \\
         &= B_1\sum_{k=1}^2\lambda_k|u_k|^2,
\end{align}
where $B_1\coloneq|\langle e_1|w\rangle|^2+|\langle e_2|w\rangle|^2$ and $u_k\coloneq\langle e_k|w\rangle/\sqrt{B_1}$.
By unitary invariance, $B_1$ has the same distribution as the sum of two of the four outcome probabilities $|w_i|^2$ of a Haar random state in $\mathbb{C}^4$. These probabilities are Haar uniform, so $B_1\sim\mathrm{Beta}(2,2)$. By a similar argument, $(u_1,u_2)$ also follows the distribution of $(\langle f_1|v\rangle,\langle f_2|v\rangle)$ for a uniform $v\in\mathbb{C}^2$ independent of $B_1$, so replacing it by this vector does not change the distribution of $P_{\ket{\psi}\bra{\psi}}$. So, we can define
\begin{equation}
Y_G(v)\coloneq\sum_{k=1}^2\lambda_k|\langle f_k|v\rangle|^2
=\bigl\|(I_{12}\otimes\bra{v}_{(3)})\ket{\Phi}\bigr\|^2 \,. 
\label{eq:Ygv-phi}
\end{equation}
Then, we obtain, for any fixed gate $G$,
\begin{equation}\label{eq:two-gate-H-avg}
\mathbb{E}_H[\varphi(P_{\ket{\psi}\bra{\psi}})]=\mathbb{E}_{B_1,v}[\varphi(B_1Y_G(v))],
\end{equation}
where $B_1\sim\mathrm{Beta}(2,2)$, and $v\in\mathbb{C}^2$ is uniformly random and independent of $G$.

\noindent Next, we will average over the gate $G$, which is applied first. Recall $\ket{\Phi} = G\ket{\psi}$ for the input state $\ket{\psi}$. So, similarly as before, we decompose $\ket{\psi}$ in terms of its Schmidt decomposition across the partition $[(1)|(2,3)]$ and then substitute it in Eq.~\eqref{eq:Ygv-phi} as follows.
\begin{equation}
Y_G(v) =\bigl\|(I_{12}\otimes\bra{v}_{(3)})\ket{\Phi}\bigr\|^2 = \lambda_0'\|(I_{2}\otimes\bra{v}_{(3)})G\ket{\chi_0}_{(2,3)}\|^2 + \lambda_1'\|(I_{2}\otimes\bra{v}_{(3)})G\ket{\chi_1}_{(2,3)}\|^2 \,,
\end{equation}
where the states $\ket{\chi_0}_{(2,3)}, \ket{\chi_1}_{(2,3)}$ occur with probability $\lambda_1',\lambda_2'$ in the decomposition of $\ket{\psi}$. Now, for fixed $t\in[0,1]$ and $v$, by Jensen's inequality,
\begin{align}
\mathbb{E}_G[\varphi\bigl(tY_G(v)\bigr)]
&\le
\sum_{j\in\{0,1\}}\lambda_j'\mathbb{E}_G[\varphi\bigl(t\bigl\|(I_{2}\otimes\bra{v}_{(3)})G\ket{\chi_j}\bigr\|^2]\bigr)\\
&= \mathbb{E}_{B_2}[\varphi(tB_2)],
\qquad B_2\sim\operatorname{Beta}(2,2).
\label{eq:two-gate-G-avg}
\end{align}
Indeed, each state $G|\chi_j\rangle$ is Haar random on $\mathbb{C}^4$, so the projections are also Haar random. So, we can combine Eq.~\eqref{eq:two-gate-H-avg} and Eq.~\eqref{eq:two-gate-G-avg}, taking the expectation over $G$ and $H$, which gives
\begin{equation}
\mathbb{E}_{G,H}[\varphi(P_\rho)] \le \mathbb{E}_{G,H}[\varphi(P_{\ket{\psi}\bra{\psi}})] \le \mathbb {E}_{B_1,B_2}\varphi(B_1B_2) \,.
\end{equation}
This proves the lemma.
\end{proof}

\bibliographystyle{alpha}
\bibliography{references}

\end{document}